\documentclass[letterpaper, 10 pt, conference]{ieeeconf}  

\IEEEoverridecommandlockouts
\usepackage{amsmath,amsfonts,amssymb}
\usepackage{mathtools}
\usepackage{graphicx}
\usepackage{epsfig}
\usepackage{xcolor}
\usepackage{accents}
\usepackage{cite}
\usepackage{comment}
\usepackage{algorithm}
\usepackage{algpseudocode}

\DeclareMathOperator{\atanh}{atanh}

\newtheorem{theorem}{Theorem}
\newtheorem{rem}{Remark}

\newtheorem{assumption}{Assumption}

\title{\LARGE \bf
Learning Input-Constrained Funnel Controllers\\ from State Trajectory Data
}

\author{{Panagiotis S. Trakas, Omid Mirzaeedodangeh, and Lars Lindemann}
\thanks{The authors are with the Automatic Control Laboratory, ETH Z\"urich, Switzerland. E-mails: {\tt\small ptrakas@ethz.ch, omirzaeedoda@ethz.ch, llindemann@ethz.ch.}}
}

\begin{document}
\bstctlcite{IEEE:BSTcontrol}
\maketitle
\thispagestyle{empty}
\pagestyle{empty}
\begin{abstract}
Designing feedback controllers that satisfy predefined performance specifications while enforcing hard input constraints is a challenging task. Our work is motivated by the idea that state trajectory data, e.g., obtained from an expert controller, often implicitly encode feasible performance attributes and input limitations. We propose an optimization-based framework that uses state trajectory data to jointly learn: (i) a performance funnel that mimics the transient and steady-state behavior encoded within the observed trajectories, and (ii) a feedback controller that enforces the learned performance specifications under hard input constraints. Unlike imitation learning methods, the proposed approach does not rely on control input data and does not reconstruct an expert policy. Instead, it synthesizes a prescribed performance controller by combining nominal model compensation with a learned state-dependent feedback gain. The resulting synthesis problem is nonconvex, for which we develop a feasibility-driven active-set synthesis procedure. Finally, we establish two complementary guarantees: a semi-global conservative actuator-authority-based certificate for prescribed performance and input satisfaction, and a local data-driven certificate ensuring these properties near sufficiently dense demonstrated trajectories.
\end{abstract}

\section{Introduction}
Prescribed performance control (PPC)~\cite{ppc,bechlaut} and the closely related funnel control (FC) methodology~\cite{bergerA} enforce transient and steady-state specifications by confining the system state or tracking error to a time-varying performance funnel. Its shape can encode requirements such as convergence rate, overshoot, and steady-state accuracy. Under hard actuator limits, however, selecting a feasible funnel \emph{a priori} is difficult: aggressive specifications may exceed the available control authority, whereas conservative ones sacrifice performance. Moreover, model-agnostic designs often rely on high feedback gains, potentially causing excessive control effort, actuator saturation, and instability through internal-signal amplification.

Actuator limitations in PPC and funnel control have recently attracted considerable research attention. In~\cite{paper1}, conservative conditions based on bounds on the system dynamics are used to manually tune the funnel and controller gains. Other approaches modify the funnel online~\cite{smc,ratesiso,bergerArch}, adapt the reference signal~\cite{fotiadis}, or introduce switching mechanisms~\cite{autom}. These methods typically preserve closed-loop stability by relaxing the original performance constraints during saturation and therefore do not guarantee them under hard input bounds. Related ideas arise in safety-critical control through control barrier functions (CBFs)~\cite{ames}, including funnel-based formulations~\cite{lanza,verginis,namerikawa}. Such approaches generally compute the control input through online optimization rather than a closed-form feedback law; moreover,~\cite{lanza,namerikawa} do not account for input constraints.

Learning-based CBF methods either learn safety certificates from expert demonstrations, often using state-input pairs~\cite{cbfl}, or synthesize neural Lyapunov/barrier certificates and controllers from model-based state samples~\cite{jin2020neural,dawson2022safe}. These approaches
enforce predefined safety and stability objectives, but do not learn
time-varying transient specifications, such as convergence rate or
overshoot, from demonstrations. In many applications, only
state trajectories are available, as in teleoperation or motion-capture
demonstrations where the corresponding commands or interaction forces
are not recorded. Such trajectories capture desirable behavior, but do
not establish whether it can be reproduced under the imposed actuator
limits. This motivates the joint synthesis of performance specifications
and a controller that can enforce them from state-only data. To the best
of our knowledge, this problem has not been addressed.

In this work, we jointly learn a performance funnel and a PPC controller from state-only trajectories. The funnel captures representative transient and steady-state behavior, whereas the controller combines nominal model compensation with a learned state-dependent gain to preserve funnel invariance under hard input constraints. Unlike methods that learn from state--input pairs~\cite{cbfl} or
estimate missing expert inputs to construct an imitation
policy~\cite{torabi2018behavioral}, our approach requires no demonstrated
input data and does not reconstruct the expert policy. We address the
resulting nonconvex synthesis problem using a feasibility-driven
procedure that iteratively computes the funnel and controller parameters
and accepts a candidate pair only after verifying funnel invariance under
actuator saturation.

The main contributions are:

\begin{itemize}
    \item A data-driven framework that jointly learns a time-varying performance funnel and an input-constrained PPC controller from state-only trajectories, with the funnel shaped by both demonstrated behavior and actuator feasibility.
    \item A feasibility-driven active-set synthesis procedure that handles the resulting nonconvex problem while explicitly accounting for input saturation.
    \item Two complementary guarantees for prescribed performance and input satisfaction: a conservative semi-global certificate based on actuator authority and a local data-driven certificate near sufficiently dense demonstrations.
\end{itemize}
\subsection{Notation}
For vectors $a,b\in\mathbb R^n$, $a\odot b$ and $a\oslash b$ denote
elementwise product and division, respectively. The vector $\mathbf{1}$ denotes a vector of ones of appropriate dimension. For a diagonal matrix
$\mathrm{diag}(v)$, $v$ denotes its diagonal. For $x\in\mathbb R^n$,
$\|x\|$ denotes the Euclidean norm. The saturation operator
$\mathrm{sat}_{\bar u}:\mathbb R^n\to\mathbb R^n$ is defined
componentwise as $[\mathrm{sat}_{\bar u}(v)]_i
= \mathrm{sign}(v_i)\min\{|v_i|,\bar u_i\},~i=1,\dots,n$.

\setlength{\parskip}{0pt}
\setlength{\textfloatsep}{4pt}
\setlength{\floatsep}{4pt}
\setlength{\intextsep}{4pt}
\setlength{\abovedisplayskip}{4pt}
\setlength{\belowdisplayskip}{4pt}
\section{Problem Formulation and Preliminaries}\label{prob}
In this section we introduce the system model, the available demonstration
data, and some preliminaries on the PPC framework used in the sequel.
Then we formally pose the main objective of this work.

\subsection{System model and demonstration data}

Consider the nonlinear control-affine system:
\begin{equation}
\dot x(t)= f(x(t)) + G(x(t))u(t,x)
\label{eq:plant}
\end{equation}
where $x(t) \in \mathbb{R}^n$ denotes the system state, $u(t,x) \in \mathbb{R}^n$ is the control input, and $f:\mathbb{R}^n \to \mathbb{R}^n$ and $G:\mathbb{R}^n \to \mathbb{R}^{n \times n}$ are locally Lipschitz continuous functions. The control input is subject to hard componentwise input constraints $|u_i(t,x)| \le \bar{u}_i$ for all $t \geq 0$, $i=1,\dots,n$, where $\bar{u}_i > 0$ are known actuator limits. Let us also define $\bar{u} := [\bar{u}_1, \dots, \bar{u}_n]^\top$.
\paragraph{State Trajectory Dataset}
We are given access to $N$ observed state trajectories $\{x^{(j)}(t)\}_{j=1}^N$, which correspond to solutions of \eqref{eq:plant} under an unknown feedback controller. These trajectories exhibit desirable performance and may arise, for example, from a human expert or an MPC controller. Our goal is to learn a closed-form feedback law that reproduces similar behavior using only state observations, without access to the corresponding control inputs. The trajectories $\{x^{(j)}(t)\}_{j=1}^N$ are collected over the time interval $[0,T]$ and sampled at discrete time instants:
\begin{equation*}
\mathcal T_{\rm grid}=\{t_k\}_{k=1}^{N_t}\subset[0,T], \quad 0 = t_1 < t_2 < \cdots < t_{N_t} = T.
\end{equation*}
The available dataset consists solely of state observations:
\begin{equation}
\mathcal{D} := \{x^{(j)}(t_k)\;:\; j=1,\dots,N,\;k=1,\dots,N_t\}.
\label{eq:dataset}
\end{equation}
In this work, we adopt a fixed time grid for clarity of presentation and to facilitate the formulation of the optimization problem. The proposed framework can be extended to non-uniform or asynchronous sampling schemes without fundamental modifications.
For notational convenience, define $x_k^{(j)} := x^{(j)}(t_k)$ for all $j=1,\dots,N,\; k=1,\dots,N_t.$
\begin{rem}
The observed trajectories encode desirable transient behavior through their temporal evolution. Therefore, meaningful funnel learning requires sufficiently rich state-space and temporal coverage. The demonstrations need not have been generated under the same actuator limits as those imposed in this paper; however, demonstrations that are compatible with the imposed limits make the synthesis problem more likely to be feasible.
\end{rem}
\subsection{Preliminaries on Prescribed Performance Control} \label{PPC}
PPC enforces transient and steady-state specifications by constraining the system state within a time-varying performance funnel. For each state component $x_i(t)$, we introduce a continuously differentiable performance function $\rho_i:[0,T]\to\mathbb{R}^\star_{+}$ and require $|x_i(t)|<\rho_i(t), ~\forall t\in[0,T], ~ i=1,\dots,n$.
The admissible state set, or \emph{performance funnel}, is defined as:
\begin{equation}
\mathcal{X}(t,\rho):=\{x\in\mathbb{R}^n:\ |x_i|<\rho_i(t),\ i=1,\dots,n\}
\label{eq:funnel}
\end{equation}
with $\rho(t):=[\rho_1(t),\dots,\rho_n(t)]^\top$. We synthesize the funnel and controller over the demonstration horizon $[0,T]$. If $\rho_i(t)$ reaches its steady-state value by $T$, setting $\rho_i(t)=\rho_i(T)$ for all $t>T$ extends the performance specification to the infinite horizon under the same modeling and operating conditions.

To transform the constraint $|x_i(t)|<\rho_i(t),~i=1,\dots,n$, into an unconstrained stabilization problem, define, for
$x\in\mathcal X(t,\rho)$, $\xi(t,x):=x\oslash\rho(t)$ and $\varepsilon(t,x):=\operatorname{atanh}(\xi(t,x))$,
where $\operatorname{atanh}$ is applied componentwise. Define also
$R(t,x):=
\operatorname{diag}\left(
\frac{1}{\rho_1(t)(1-\xi_1^2(t,x))},\ldots,
\frac{1}{\rho_n(t)(1-\xi_n^2(t,x))}
\right)$ and $s(t,x):=R(t,x)^\top\varepsilon(t,x).$
Along a trajectory $x(t)$, we use the shorthand
$\xi(t):=\xi(t,x(t))$, $\varepsilon(t):=\varepsilon(t,x(t))$,
$R(t):=R(t,x(t))$, and $s(t):=s(t,x(t))$. Differentiating
$\varepsilon(t)$ along the system trajectories yields:
\begin{equation}
\dot{\varepsilon}(t)
=
R(t)\bigl(f(x(t))+G(x(t))u(t,x(t))-q(t,x(t))\bigr)
\label{eq:epsdot_vector}
\end{equation}
where $q(t,x):=
x\odot\bigl(\dot{\rho}(t)\oslash\rho(t)\bigr)$ accounts for time-
varying funnel dynamics. Since
$\xi=\tanh(\varepsilon)$ componentwise, boundedness of
$\varepsilon(t)$ implies that $\xi(t)$ remains strictly inside
$(-1,1)^n$ and hence that $|x_i(t)|<\rho_i(t)$. Accordingly, PPC
typically seeks to establish boundedness of $\varepsilon$, for example,
through a Lyapunov function defined in the transformed coordinates.
Classical PPC~\cite{bechlaut} often employs feedback of the form
$u=-Ks$, with a constant gain matrix $K\succ0$. Without additional
mechanisms or assumptions, however, such feedback does not ensure satisfaction of hard input constraints. In this paper, we combine nominal model compensation with a learned PPC feedback gain to enforce the prescribed performance specifications under hard actuator limits.
\subsection{Working assumption}
\begin{assumption}
\label{ass:model}
There exist known and locally Lipschitz continuous nominal models $\hat{f}:\mathbb{R}^n \to \mathbb{R}^n$ and $\hat{G}:\mathbb{R}^n \to \mathbb{R}^{n \times n}$, known constants $\bar{\delta}_f, \bar{\delta}_G > 0$, and a known compact set $\mathcal X_c \subset \mathbb{R}^n$ such that:
\begin{equation}
f(x) = \hat{f}(x) + \Delta_f(x), \quad G(x) = \hat{G}(x) + \Delta_G(x)
\end{equation}
with $\|\Delta_f(x)\| \le \bar{\delta}_f$ and $\|\Delta_G(x)\| \le \bar{\delta}_G$ for all $x \in \mathcal X_c$. Moreover, the nominal input matrix $\hat G(x)$ is uniformly nonsingular on $\mathcal X_c$, i.e., there exists $\underline g > 0$ such that $\|\hat G(x)^{-1}\| \le \underline g^{-1},~ \forall x \in \mathcal X_c$. 
\end{assumption}

The set $\mathcal X_c$ is assumed to contain all sampled demonstrated
states and every candidate funnel considered during synthesis, i.e.,
$x_k^{(j)}\in\mathcal X_c$ for all $j=1,\dots,N$ and
$k=1,\dots,N_t$, and
$\mathcal X(t,\rho)\subseteq\mathcal X_c$ for all $t\in[0,T]$ and every candidate performance function $\rho(t)$. Since $\hat f$ is continuous on the compact set $\mathcal X_c$, define $\bar F_{\hat f}:=\sup_{x\in\mathcal X_c}\|\hat f(x)\|$.

In contrast to model-free PPC/FC designs \cite{bergerArch,fotiadis,autom,smc,ratesiso}, we exploit nominal model knowledge together with state trajectory data to enforce prescribed performance under input constraints. Rather than relaxing specifications online, we search for a funnel-controller pair compatible with actuator limits and verify feasibility a posteriori. Moreover, unlike classical PPC/FC assumptions requiring $\mathrm{Sym}(G(x)) \succeq \underline g I$~\cite{ppc,bergerArch,fotiadis,autom,smc,ratesiso,lanza,namerikawa,paper1}, we adopt the milder condition of uniform nonsingularity of $\hat G(x)$.
\begin{rem}
The availability of a nominal model on an operational region $\mathcal X_c$ is not restrictive in practice. For instance, in robotics and autonomous systems, approximate models are typically available from offline identification and can be refined using approximators, e.g., neural networks. Since our approach only requires bounded model uncertainty, it remains applicable under moderate modeling errors.
\end{rem}
\subsection{Control objective}

Given the system \eqref{eq:plant} and the state-only dataset $\mathcal{D}$, we aim to jointly synthesize a performance function $\rho(t)$ and a state-feedback controller $u(t,x)$ such that:

\begin{enumerate}
\item \textbf{Partial trajectory coverage:} for each state component $i=1,\dots,n$ and grid point $t_k$, the learned funnel boundary $\rho_i(t_k)$ encloses at least a prescribed fraction of the demonstrated values $\{|x_i^{(j)}(t_k)|\}_{j=1}^N$. This coverage mechanism is illustrated
in Fig.~\ref{fig:funnel_quantiles} and formalized through the empirical
quantile bounds introduced in Section~\ref{sec:funnel_param}.

\item \textbf{Prescribed performance:} for all initial conditions satisfying $|x_i(0)| < \rho_i(0)$, $i=1,\dots,n$, the closed-loop trajectory satisfies:
\begin{align}
    |x_i(t)|<\rho_i(t), ~\forall t\in[0,T], ~ i=1,\dots,n\label{eq:ppc_bounds}.
\end{align}
\vspace{-2mm}
\item \textbf{Input constraint satisfaction:} The control input satisfies:
\begin{equation}
|u_i(t,x)|\leq\bar u_i,
~
\forall\,t\in[0,T],~
i=1,\dots,n
\label{eq:input_constraints}
\end{equation}
for all $x\in\mathcal X(t,\rho).$
\end{enumerate}

\begin{rem}
Prescribed performance is the primary objective, while demonstration coverage and input feasibility constrain the funnel and controller design. The prescribed coverage level regulates the fraction of demonstrations enclosed by the funnel at each sampling instant, thereby providing robustness to outliers.
\end{rem}
\section{Funnel and Controller Parameterization with Design Constraints}
\label{sec:param_constraints}

We parameterize the performance funnel and the PPC controller and introduce the constraints required for their joint synthesis, which we present thereafter.

\subsection{Funnel Parameterization and Constraints}
\label{sec:funnel_param}

The performance function $\rho(t;\theta)$ defines the funnel \eqref{eq:funnel} and is parameterized componentwise using cubic B-splines:
\begin{equation}
\rho_i(t;\theta) = \theta_i^\top \Psi(t),
\quad \theta_i \in \mathbb{R}^{n_c}, \quad i=1,\dots,n
\label{eq:bspline_param}
\end{equation}
where $\Psi(t):= [\psi_1(t), \dots, \psi_{n_c}(t)]^\top$ is a vector of fixed $C^2$ basis functions. The number of basis functions $n_c$ regulates the trade-off between expressivity and regularity.
\begin{rem}
Cubic B-splines are used due to their $C^2$ smoothness and linear parameterization, which yields closed-form expressions for $\dot{\rho}_i(t;\theta)$ and $\ddot{\rho}_i(t;\theta)$ while preserving convexity in the parameters~\cite{splines}. Moreover, the B-spline basis functions are nonnegative and form a partition of unity, i.e., $\psi_j(t)\ge 0$ and $\sum_j \psi_j(t)=1$ for all $t\in[0,T]$, which is key later for deriving continuous-time guarantees such as positivity and boundedness of the funnel.
\end{rem}
\paragraph{Funnel Positivity Constraint}
Using the nonnegativity and partition-of-unity of the B-spline basis, we impose:
\begin{equation}
\theta_i \succeq \underline{\rho}_i \mathbf{1},
\quad i=1,\dots,n
\label{eq:theta_positive}
\end{equation}
where $\underline{\rho}_i$ is a design parameter that guarantees $\rho_i(t;\theta)\ge \underline{\rho}_i>0$.

\paragraph{Data-Driven Bounds}
For every state component $i=1,\dots,n$ and grid index
$k=1,\dots,N_t$, we constrain the learned funnel boundary to a
data-driven feasible interval:
\begin{equation}
\rho_i^{\mathrm{low}}(t_k)
\leq
\rho_i(t_k;\theta)
\leq
\rho_i^{\mathrm{up}}(t_k).
\label{eq:funnel_bounds}
\end{equation}
The envelopes $\rho_i^{\mathrm{low}}(t_k)$ and
$\rho_i^{\mathrm{up}}(t_k)$ are constructed using empirical quantiles
rather than the maximum value of expert trajectories. This reduces sensitivity to a small number of extreme state values and yields bounds that better represent the dominant
transient and steady-state behavior of the demonstrations.

To define the empirical quantiles, let $x_{i,k}^{[1]}\leq\cdots\leq x_{i,k}^{[N]}$
denote the order statistics obtained by sorting the absolute values of the $i$-th state
component across the $N$ demonstrations at time $t_k$, i.e.,
$\{|x_i^{(j)}(t_k)|\}_{j=1}^N$, with repeated values retained. For
$\gamma\in(0,1]$, define:
\begin{equation}
Q_{\gamma}\!\left(\{|x_i^{(j)}(t_k)|\}_{j=1}^N\right)
:=
x_{i,k}^{[\lceil\gamma N\rceil]}.
\label{eq:empirical_quantile}
\end{equation}
Thus, $Q_\gamma$ is the smallest demonstrated magnitude such that at
least a fraction $\gamma$ of the demonstrations do not exceed it.

Next, choose $0<\alpha\leq\beta\leq1$ and define:
\begin{align*}
\rho_i^{\mathrm{up}}(t_k)
&:=
Q_{\beta}\!\left(\{|x_i^{(j)}(t_k)|\}_{j=1}^N\right)
\\
\rho_i^{\mathrm{low}}(t_k)
&:=
\omega_{i,k}\rho_i^{\mathrm{up}}(t_k)
+(1-\omega_{i,k})
Q_{\alpha}\!\left(\{|x_i^{(j)}(t_k)|\}_{j=1}^N\right)
\end{align*}
where $\omega_{i,k}:=e^{-\lambda_i t_k}$ and $\lambda_i>0$. Since
$t_1=0$, we have $\omega_{i,1}=1$ and hence
$\rho_i^{\mathrm{low}}(t_1)=\rho_i^{\mathrm{up}}(t_1)$. As time
increases, $\omega_{i,k}$ decreases, and the lower envelope approaches
the tighter $\alpha$-quantile envelope. The parameter $\lambda_i$
controls the rate of this transition. For fixed $\alpha$, $\beta$, and $\lambda_i$, the envelopes are
precomputed. Since $\rho_i(t_k;\theta)$ is linear in $\theta$,
\eqref{eq:funnel_bounds} defines linear constraints on $\theta$. In
practice, the envelopes may be smoothed, e.g., using a moving-average
filter, before enforcing \eqref{eq:funnel_bounds} to reduce
grid-to-grid variations. This quantile-based approach is illustrated in Fig.~\ref{fig:funnel_quantiles}.
\begin{figure}[t]
\centering
\includegraphics[width=\linewidth,trim={0.2cm 0.2cm 0.2cm 0.2cm},clip]{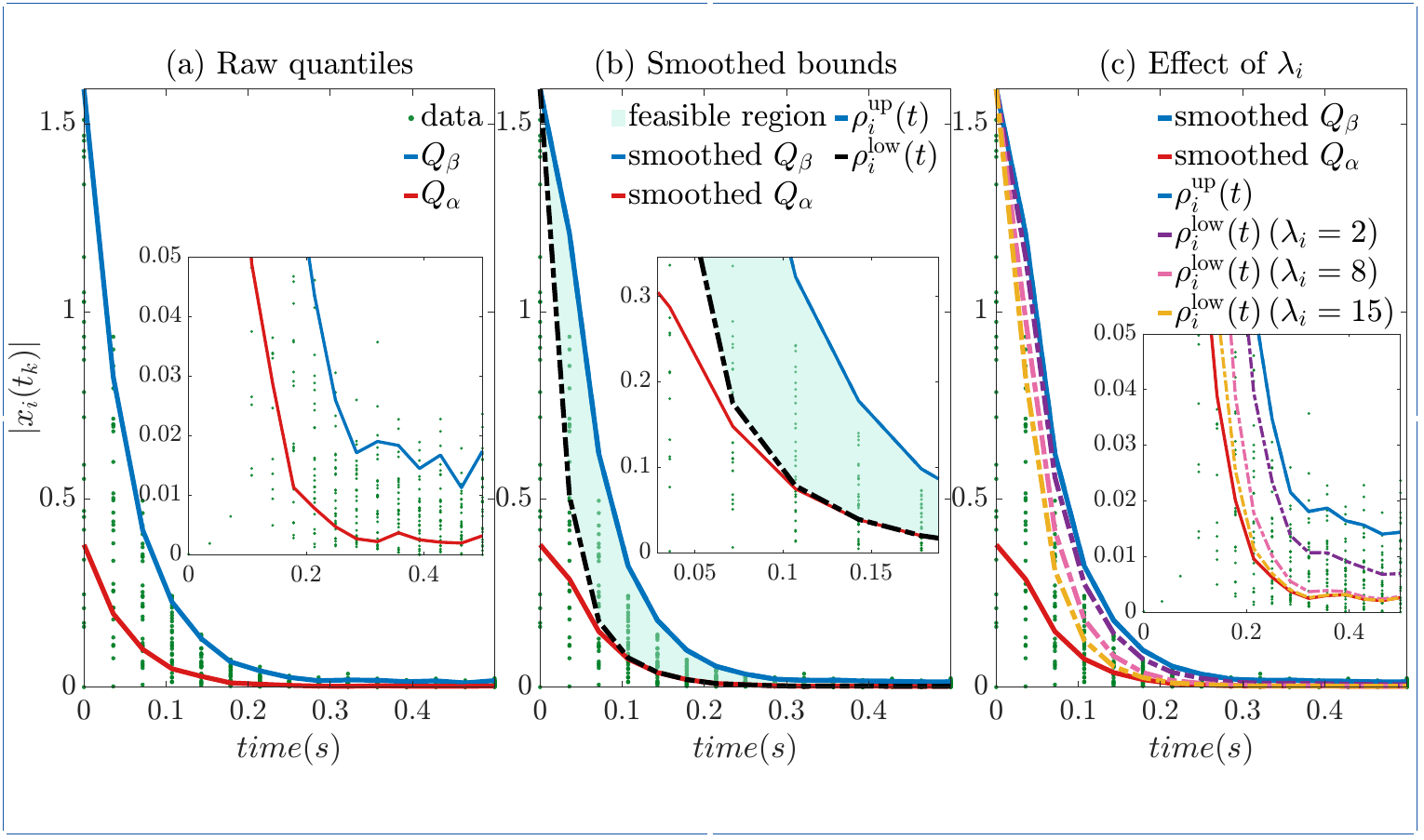}
\caption{Data-driven funnel bounds. Left: demonstrations and empirical
quantiles ($\alpha=0.2$, $\beta=0.9$). Middle: bounds smoothed using a
moving-average window of size two and the resulting feasible interval
$[\rho_i^{\mathrm{low}}(t),\rho_i^{\mathrm{up}}(t)]$. Right: effect of
$\lambda_i$ on $\rho_i^{\mathrm{low}}(t)$.}
\label{fig:funnel_quantiles}
\end{figure}
\begin{rem}
The envelopes $\rho_i^{\mathrm{up}}(t_k)$ and $\rho_i^{\mathrm{low}}(t_k)$ are
design bounds derived from the demonstrations, not estimates of an
underlying expert funnel.
\end{rem}

\paragraph{Funnel Dynamics Constraints}
Let $h:=\max \limits _{k=1,\dots,N_t-1}(t_{k+1}-t_k)$ denote the maximum grid spacing. Choose $v\in(0,1]$ and $c_i>0$ such that for all $i=1,\dots,n$ it holds $c_i\leq\frac{v\lambda_i}{h}$. To limit oscillations and control the maximum contraction rate, we impose, for $i=1,\dots,n$ and $k=1,\dots,N_t$:
\begin{align}
-(v\lambda_i-c_i h)
&\leq \dot{\rho}_i(t_k;\theta)
\leq v\lambda_i-c_i h
\label{eq:funnel_rate}\\
-c_i
&\leq \ddot{\rho}_i(t_k;\theta)
\leq c_i.
\label{eq:funnel_curv}
\end{align}

\begin{rem}
The parameter $\lambda_i$ determines the nominal contraction rate of
$\rho_i^{\mathrm{low}}(t)$, while $v$ scales the admissible rate of the learned
funnel. In implementation, $v$ is selected from a finite candidate set. The margin $c_i h$ accounts for the maximum change in
$\dot{\rho}_i$ between discrete time points, ensuring the continuous-time bound
$|\dot{\rho}_i(t;\theta)|\leq v\lambda_i$, as formalized
below.
\end{rem}

The following result shows that these constraints guarantee funnel
positivity and uniform continuous-time bounds on $\rho_i$,
$\dot{\rho}_i$, and $\ddot{\rho}_i$ over $[0,T]$.
\begin{theorem}
\label{lem:slope_uniform}
For each $i=1,\dots,n$, let
$\rho_i(\cdot;\theta):[0,T]\to\mathbb R$ be a $C^2$ cubic B-spline
that is a cubic polynomial on every interval $[t_k,t_{k+1}]$.
If \eqref{eq:funnel_curv} holds
at every discrete time $t_k$, $k=1,\dots,N_t$,
then $|\ddot\rho_i(t;\theta)|\leq c_i$ for all $t\in[0,T]$.
If, in addition, \eqref{eq:funnel_rate} holds
at every $t_k$, $k=1,\dots,N_t$,
then $|\dot\rho_i(t;\theta)|\leq v\lambda_i$ for all
$t\in[0,T]$. If, additionally, \eqref{eq:theta_positive} and
\eqref{eq:funnel_bounds} hold at every $t_k$, then, for every
$k=1,\dots,N_t-1$ and $t\in[t_k,t_{k+1}]$:
\[
0<\underline{\rho}_i
\leq\rho_i(t;\theta)
\leq\rho_i^{\mathrm{up}}(t_k)+v\lambda_i h.
\]
\end{theorem}

\begin{proof}
Since $\rho_i$ is a cubic spline that restricts to a cubic polynomial
on each interval $[t_k,t_{k+1}]$, $\ddot{\rho}_i$ is linear on that
interval. Therefore, $|\ddot{\rho}_i|$ attains its maximum at one of the
endpoints. Since \eqref{eq:funnel_curv} holds at every $t_k$, it follows
that $|\ddot{\rho}_i(t;\theta)|\leq c_i$ for all $t\in[0,T]$. Next, for $t\in[t_k,t_{k+1}]$, the fundamental theorem of calculus gives $\dot\rho_i(t;\theta)-\dot\rho_i(t_k;\theta)=\int_{t_k}^t \ddot\rho_i(\tau;\theta)\,d\tau$,
hence
$|\dot\rho_i(t;\theta)-\dot\rho_i(t_k;\theta)|
\leq c_i|t-t_k|\leq c_i h$. Combining with $|\dot\rho_i(t_k)|\le v\lambda_i-c_i h$ yields $|\dot\rho_i(t;\theta)|\le v\lambda_i$. Similarly, for $t\in[t_k,t_{k+1}]$, we have: $\rho_i(t;\theta)\le \rho_i(t_k;\theta)+v\lambda_i|t-t_k|
\le \rho_i^{\mathrm{up}}(t_k)+v\lambda_i h$ for any $t\in[t_k,t_{k+1}]$. Finally, since $\rho_i(t;\theta)=\theta_i^\top\Psi(t)$ with $\theta_i \succeq \underline\rho_i\mathbf{1}$ and the B-spline basis satisfies $\Psi(t)\succeq0$ and $\sum_{j=1}^{n_c}\psi_j(t)=1$, we obtain $\rho_i(t;\theta)
=
\sum_{j=1}^{n_c}\theta_{ij}\psi_j(t)
\ge
\underline\rho_i \sum_{j=1}^{n_c}\psi_j(t)
=
\underline\rho_i>0$ for all $t\in[0,T]$.
\end{proof}
\subsection{Controller Parameterization and Constraints}
\label{sec:controller_param}
\vspace{-1mm}
We consider the following PPC controller with nominal
model compensation:
\begin{equation}
u(t,x;\theta,w)
=
\mathrm{sat}_{\bar u}\!\left(
u_d(t,x;\theta,w)
\right)
\label{eq:controller}
\end{equation}
where the desired input is given by:
\begin{align}
    u_d(t,x;\theta,w):=\hat G(x)^{-1}\bigl[u_{\mathrm{n}}(t,x;\theta)-K(\xi;w)s \bigr]\label{desiredu}
\end{align}
with $u_{\mathrm{n}}(t,x;\theta)
:= q(t,x;\theta)-\hat f(x)
$.
The PPC signals $\xi$, $s$, and $q$ depend on $\rho(t;\theta)$ as defined in Section~\ref{PPC}. For clarity, this dependence is omitted in the sequel. The term $u_{\mathrm{n}}$ compensates the nominal drift and funnel dynamics, while the feedback term $-Ks$ enforces prescribed performance and provides robustness to model uncertainty. The saturation operator is used to ensure that the control input satisfies the actuator constraints $|u_i(t,x)|\le \bar u_i$ for all $t$.

\paragraph{Gain Parameterization} 
To match the componentwise structure of the prescribed performance constraints and maintain a low-complexity synthesis problem, we employ a diagonal gain matrix:
\begin{equation}
K(\xi;w) = \mathrm{diag}\bigl(k_1(\xi_1;w_1), \dots, k_n(\xi_n;w_n)\bigr).
\label{eq:gain_diag}
\end{equation}
Each gain is parameterized using radial basis functions (RBFs):
\begin{equation}
k_i(\xi_i;w_i)
=
k_{\min}+w_i^\top \phi_i(\xi_i), \qquad w_i \succeq 0
\label{eq:gain_param}
\end{equation}
where $\phi_i(\xi_i)
=
[\varphi_{i1}(\xi_i),\dots,\varphi_{iN_K}(\xi_i)]^\top
\in\mathbb R^{N_K}$ and $\varphi_{i\ell}(\xi_i)
=
\exp\!\left(
-\frac{(\xi_i-c_{i\ell})^2}{2\sigma_{i\ell}^2}
\right),
$ $i=1,\dots,n$, $\ell=1,\dots,N_K$,
with fixed centers $c_{i\ell}\in[-1,1]$ and widths
$\sigma_{i\ell}>0$, as well as a design parameter $k_{\min}>0$ that enforces a strictly positive baseline gain, ensuring a minimum level of feedback action required for stability analysis (cf. Section~\ref{sec:baseline_cert}). The condition $w_i \succeq 0$ ensures $k_i(\xi_i;w_i)\ge k_{\min}>0$ and restricts the parameterization to nonnegative combinations of basis functions, which simplifies the synthesis and ensures strictly positive gains. Hence, $K(\xi;w)\succeq k_{\min}I$. Let us also define the stacked parameter vector as $w := [w_1^\top,\dots,w_n^\top]^\top$.
\begin{rem}
RBFs enable state-dependent gain shaping across the funnel while preserving affine dependence on $w_i$ for fixed $\theta$. The parameters $N_K$, $c_{i\ell}$, and $\sigma_{i\ell}$ balance expressiveness and complexity; in practice, few uniformly spaced centers over $[-1,1]$ suffice. Alternative parameterizations (e.g., neural networks) could also be employed, but would generally introduce nonlinear constraints in the resulting optimization problem.
\end{rem}

\paragraph{Sampled Lyapunov Decrease Conditions}
Since boundedness of $\varepsilon(t)$ implies satisfaction of \eqref{eq:ppc_bounds}, consider the Lyapunov function candidate $V(\varepsilon)=\tfrac12\varepsilon^\top\varepsilon$. Along the closed-loop system we have $\dot V
=
s^\top\bigl(f(x)+G(x)u-q(t,x;\theta)\bigr)$.
Using \eqref{eq:controller} and the decompositions $f=\hat f+\Delta_f$ and $G=\hat G+\Delta_G$ defined in Assumption \ref{ass:model} as well as invoking the bounds on the model mismatch together with $\|u\|\le \|\bar u\|$, we obtain:
\begin{equation}
\dot V
\le
s^\top\bigl(\hat f(x)+\hat G(x)\mathrm{sat}_{\bar u}(u_d)-q(t,x;\theta)\bigr)
+
\|s\|\bigl(\bar\delta_f+\bar\delta_G\|\bar u\|\bigr)
\label{eq:Vdot_bound_sat}
\end{equation}
The bound in \eqref{eq:Vdot_bound_sat} holds for all
$t\in[0,T]$ and $x\in\mathcal X(t,\rho)$.
To reduce conservatism, we enforce strict Lyapunov decrease only at
demonstrated samples in the boundary region $\|s(t,x)\|\geq\delta$,
shown in red in Fig.~\ref{fig:lyap_activation}.

To this end, for each demonstration point $x_k^{(j)}$ lying strictly inside the candidate funnel, i.e., satisfying
$|x_i^{(j)}(t_k)|<\rho_i(t_k;\theta)$ for all $i$, define $\xi_k^{(j)} := x_k^{(j)} \oslash \rho(t_k;\theta)$,
$\varepsilon_k^{(j)} := \mathrm{atanh}(\xi_k^{(j)})$, 
$s_k^{(j)} := R_k^{(j)\top}\varepsilon_k^{(j)}$ with $R_k^{(j)} := R(t_k,x_k^{(j)})$. Moreover, define $u_{d_k}^{(j)} := u_d(t_k,x_k^{(j)};\theta,w)$, $A_k^{(j)} := s_k^{(j)\top}\hat G(x_k^{(j)})$ and $B_k^{(j)} := s_k^{(j)\top}\bigl(\hat f(x_k^{(j)})-q(t_k,x_k^{(j)};\theta)\bigr) +\|s_k^{(j)}\|\bigl(\bar\delta_f+\bar\delta_G\|\bar u\|\bigr)$. We enforce Lyapunov decrease selectively using the indicator function $b_\delta(\chi) := \mathbf{1}\{\chi \ge \delta\}$, which activates near the funnel boundary (i.e., for large $\|s\|$). Specifically, we impose:
\begin{equation}
b_\delta(\|s_k^{(j)}\|)
\Bigl(
A_k^{(j)}\,\mathrm{sat}_{\bar u}(u_{d_k}^{(j)})
+B_k^{(j)}
\Bigr)
\leq
-\mu\,b_\delta(\|s_k^{(j)}\|)
\label{eq:lyap_constraint}
\end{equation}
with $\mu>0$. Thus, \eqref{eq:lyap_constraint} enforces
$\dot V(\varepsilon_k^{(j)})\leq-\mu$ at active samples and is
inactive otherwise; see Fig.~\ref{fig:lyap_activation}.
\begin{figure}[t]
\centering
\includegraphics[width=1\linewidth,trim={0cm 0cm 0cm 0cm},clip]{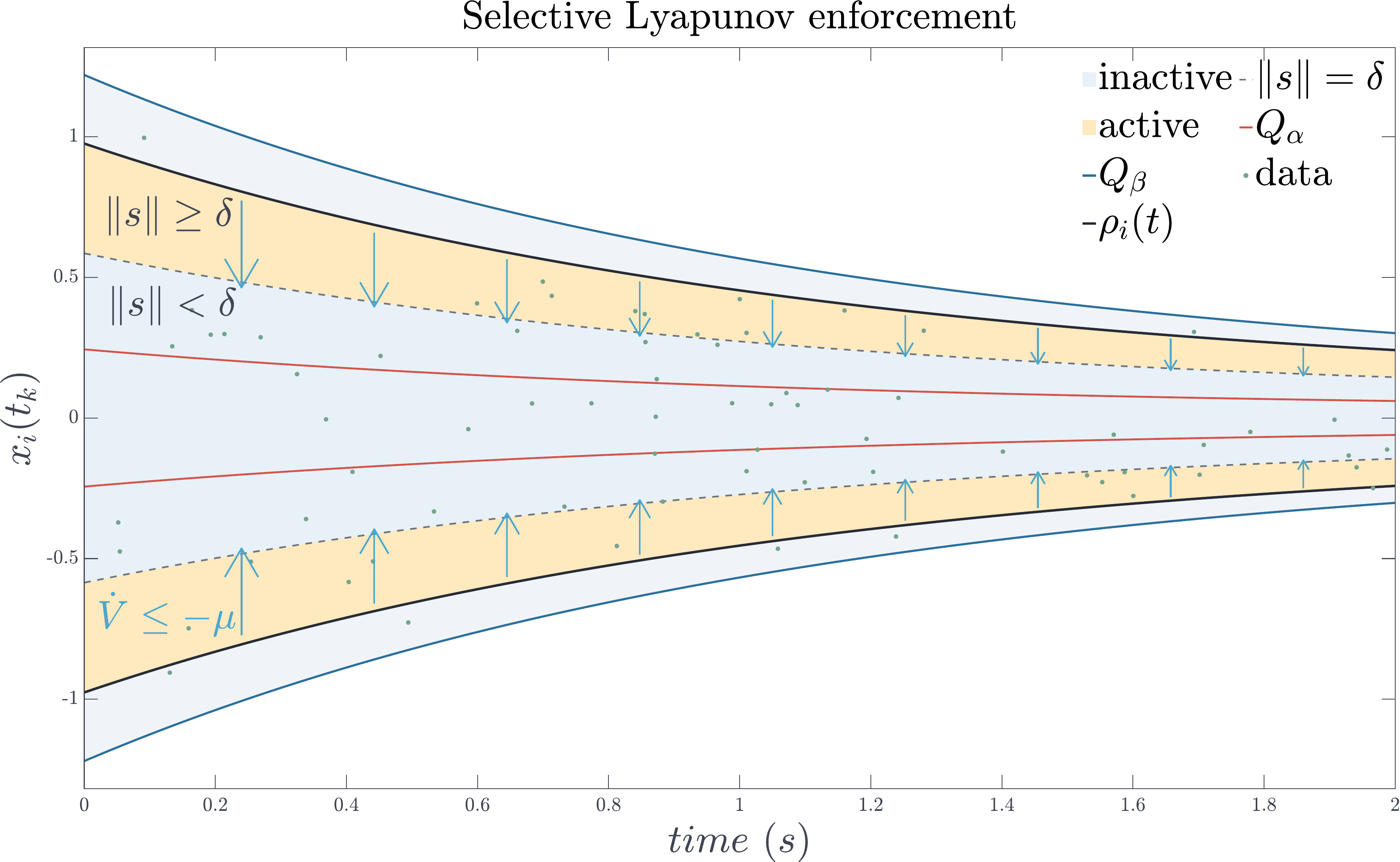}
\caption{Selective Lyapunov enforcement. The funnel $\rho_i(t)$ lies within the quantile envelopes $Q_\alpha$ and $Q_\beta$; $\|s\|=\delta$ splits the inactive ($\|s\|<\delta$) and active ($\|s\|\ge\delta$) regions, where arrows mark inward motion $\dot V\le-\mu$.}

\label{fig:lyap_activation}
\end{figure}

\section{Data-Driven Funnel and Controller Synthesis}\label{learning}
\subsection{Joint Optimization Problem}
\label{sec:optimization}
Building on the funnel and controller parameterizations in Section~\ref{sec:param_constraints}, we formulate the following data-driven joint synthesis problem:
\begin{subequations}
\label{eq:joint_problem}
\begin{align}
\min_{\theta, w} \quad 
& 
\sum_{k=1}^{N_t}\sum_{i=1}^n \rho_i(t_k;\theta)
+
\ell_w \|w\|^2
\label{eq:objective} \\
\text{s.t.}  \quad & 
\rho_i^{\mathrm{low}}(t_k)
\le
\rho_i(t_k;\theta)
\le
\rho_i^{\mathrm{up}}(t_k)
\label{eq:c1}
\\
& -(v\lambda_i - c_i h)
\le
\dot{\rho}_i(t_k;\theta)
\le
v\lambda_i - c_i h
\label{eq:c2}
\\
& -c_i
\le
\ddot{\rho}_i(t_k;\theta)
\le
c_i
\label{eq:c3}
\\
& \theta_i \succeq \underline{\rho}_i \mathbf{1}
\label{eq:c4}
\\
& b_\delta(\|s_k^{(j)}\|)\Bigl(
A_k^{(j)}\,\mathrm{sat}_{\bar u}(u_{d_k}^{(j)}) + B_k^{(j)}
\Bigr)
\leq
-\mu\,b_\delta(\|s_k^{(j)}\|)
\label{eq:c5}
\\
& w_i \succeq 0.
\label{eq:c6}
\end{align}
\end{subequations}
Constraints \eqref{eq:c1}-\eqref{eq:c3} are imposed for all
$i=1,\dots,n$ and $k=1,\dots,N_t$, \eqref{eq:c4} and
\eqref{eq:c6} for all $i=1,\dots,n$, and \eqref{eq:c5} for every
$(j,k)$ such that
$x_k^{(j)}\in\mathcal X(t_k,\rho(\cdot;\theta))$.
 The objective \eqref{eq:objective} promotes tight funnels while penalizing large feedback gains. The scalar $\ell_w>0$ weights the gain regularization. Constraints \eqref{eq:c1}-\eqref{eq:c4} define the feasible funnel set. Constraint~\eqref{eq:c5} is imposed for every $j,k$ such that
$|x_i^{(j)}(t_k)|<\rho_i(t_k;\theta)$ for all $i=1,\dots,n$, while \eqref{eq:c6} ensures positivity of the learned gains. Input feasibility follows directly from the saturation. Thus, a feasible solution of \eqref{eq:joint_problem} defines a candidate funnel-controller pair whose invariance condition is verified only at
the sampled points. Section~\ref{sec:guarantees} provides conditions
under which this certificate extends over continuous time and state
space. The main synthesis parameters affect the trade-off between performance, robustness, and feasibility as summarized in Table~\ref{tab:tuning}.
\begin{table}[t]
\caption{Main synthesis parameters and their roles.}
\label{tab:tuning}
\centering
\scriptsize
\setlength{\tabcolsep}{2pt}
\renewcommand{\arraystretch}{1.05}
\begin{tabular}{@{}p{0.23\linewidth}p{0.73\linewidth}@{}}
\hline
Parameter & Role and effect\\
\hline
$\alpha,\beta$
& Set the lower and upper quantile envelopes, regulating demonstration coverage and funnel tightness.\\
$\lambda_i$
& Sets the envelope convergence rate; larger values permit faster contraction but may require greater control authority.\\
$v,\gamma_v,v_{\min}$
& $v$ scales the admissible contraction rate, while $\gamma_v$ and $v_{\min}$ define the candidate search.\\
$c_i$ &
Bounds the funnel curvature $|\ddot{\rho}_i|$ and affects smoothness and feasibility via $c_i h$ in \eqref{eq:c2}. \\
$n_c,\underline{\rho}_i$
& Set the spline resolution and minimum funnel width.\\
$N_K,c_{i\ell},\sigma_{i\ell}$
& Determine the resolution and localization of the RBF gain parameterization.\\
$k_{\min},\ell_w$
& Set the baseline gain and regularization of the learned gain weights.\\
$\mu,\delta$
& Set the Lyapunov-decrease margin and activation threshold; larger $\mu$ strengthens decrease, whereas larger $\delta$ restricts enforcement closer to the boundary; see Fig \ref{fig:lyap_activation}.\\
\hline
\end{tabular}
\end{table}
\subsection{Feasibility-Driven Active-Set Synthesis}

The
optimization problem \eqref{eq:joint_problem} is nonconvex in general.
However, constraints \eqref{eq:c1}-\eqref{eq:c4} are linear in $\theta$.
For fixed $\theta$, the quantities $s_k^{(j)}$, $A_k^{(j)}$, and
$B_k^{(j)}$ are fixed, while $u_{d_k}^{(j)}$ is affine in $w$.
Thus, once $\theta$ is fixed, saturation is the only
source of nonconvexity in the controller subproblem.

We exploit this structure through a feasibility-driven active-set procedure that solves a linear funnel problem followed by a sequence of convex controller subproblems with fixed saturation patterns. The synthesis uses two nested loops. 

\emph{Outer loop (funnel synthesis):}
Initialize $v_0=v$. At outer iteration $q$, solve the funnel problem
\eqref{eq:c1}--\eqref{eq:c4} for the current $v_q$. The resulting
$\theta_q^\star$ and $v_q$ are held fixed during the subsequent controller
synthesis.

\emph{Inner loop (controller synthesis):}
Initialize $w_0:=\mathbf 0$, corresponding to the baseline gain
$K(\xi;w_0)=k_{\min}I$. At active-set iteration
$r\in\mathbb N_0$, the current gain $w_r$ induces, for every
$j=1,\dots,N$ and $k=1,\dots,N_t$ satisfying $x_k^{(j)}\in\mathcal X\bigl(\rho(t_k;\theta_q^\star)\bigr),$
the index sets:
\[
\begin{aligned}
\mathcal F_{k,r}^{(j)}
&:=
\left\{
i\in\{1,\dots,n\}:
|u_{d_k,i}^{(j)}(w_r)|\leq\bar u_i
\right\},\\
\mathcal S_{k,r}^{(j)}
&:=
\left\{
i\in\{1,\dots,n\}:
|u_{d_k,i}^{(j)}(w_r)|>\bar u_i
\right\}.
\end{aligned}
\]
These sets and the signs of the saturated components define the current
saturation pattern. Holding this pattern fixed, define:
\begin{equation}
\tilde u_{k,i}^{(j)}(w)=
\begin{cases}
u_{d_k,i}^{(j)}(w),
& i\in\mathcal F_{k,r}^{(j)},\\
\operatorname{sign}\!\left(
u_{d_k,i}^{(j)}(w_{r})
\right)\bar u_i,
& i\in\mathcal S_{k,r}^{(j)}.
\end{cases}
\label{eq:sat_pattern}
\end{equation}
Let $\tilde u_k^{(j)}(w)
:=
\bigl[
\tilde u_{k,1}^{(j)}(w),\dots,
\tilde u_{k,n}^{(j)}(w)
\bigr]^\top$. Replacing the saturation map in \eqref{eq:c5} by the affine branch
associated with the current pattern gives:
\begin{equation}
b_\delta(\|s_k^{(j)}\|)
\left(
A_k^{(j)}\tilde u_k^{(j)}(w)+B_k^{(j)}
\right)
\leq
-\mu\,b_\delta(\|s_k^{(j)}\|).
\label{eq:c5_active}
\end{equation}
This constraint is exact for any $w$ that preserves, relative to
$w_r$, whether each input component is unsaturated or saturated. In that case, $\tilde u_k^{(j)}(w)
=
\operatorname{sat}_{\bar u}\!\left(u_{d_k}^{(j)}(w)\right)$, and \eqref{eq:c5_active} coincides with \eqref{eq:c5}.

At iteration $r$, minimizing $\ell_w\|w\|^2$ subject
to \eqref{eq:c6} and \eqref{eq:c5_active} yields $w_{r+1}$. The
original constraint \eqref{eq:c5} is then evaluated at all considered
samples using the actual saturated input
$\operatorname{sat}_{\bar u}(u_{d_k}^{(j)}(w_{r+1}))$. If it holds,
$(\theta_q^\star,w_{r+1})$ is accepted; otherwise, the saturation pattern
is updated and the inner loop continues.

\emph{Outer update and termination:}
The inner loop for the current $v_q$ terminates without a verified
candidate if the controller subproblem has no solution, the pattern
stops changing while \eqref{eq:c5} remains violated, or the iteration
limit is reached. If $v_q>v_{\min}$, set:
\[
v_{q+1}
=
\max\{v_{\min},\gamma_v v_q\},
\qquad \gamma_v\in(0,1)
\]
and repeat the funnel and controller steps. The procedure terminates
after $v_q=v_{\min}$ has been tested. The value $v_{\min}$ is chosen
such that $c_i h\leq v_{\min}\lambda_i$ for all $i$. Reducing $v_q$
restricts the allowable funnel rate and can reduce the required control
authority; see Theorem~\ref{lem:ppc_certificate}.
Algorithm~\ref{alg:alt} summarizes the procedure.

\begin{algorithm}[t]
\caption{Feasibility-driven active-set synthesis}
\label{alg:alt}
\begin{algorithmic}[1]
\State \textbf{Input:}
$\mathcal D$, $\alpha,\beta$, $\lambda_i,c_i,\underline\rho_i$,
$\ell_w,\mu,\delta$
\State \textbf{Choose:}
$v\in(0,1]$, $\gamma_v\in(0,1)$,
$v_{\min}\in(0,v]$, and
$r_{\max}\in\mathbb N_{>0}$, with
$c_i h\leq v_{\min}\lambda_i$ for all $i=1,\dots,n$
\State Generate $v_0,\dots,v_Q$ using
$v_0=v$ and
$v_{q+1}=\max\{v_{\min},\gamma_vv_q\}$
until $v_Q=v_{\min}$

\For{$q=0,\dots,Q$}
    \State Solve the funnel linear program
    \Statex \hspace{\algorithmicindent}
    $\displaystyle
    \theta_q^\star\in
    \arg\min_\theta\sum_{k,i}\rho_i(t_k;\theta)
    \quad
    \text{s.t. }\eqref{eq:c1}\text{--}\eqref{eq:c4}$
    \If{\eqref{eq:c1}--\eqref{eq:c4} admit no solution}
        \State \textbf{break}
    \EndIf

    \State Define
    $\mathcal I_q\gets
    \{(j,k):x_k^{(j)}
    \in\mathcal X(t,\rho(\cdot;\theta_q^\star))\}$
    \State Compute $s_k^{(j)}$, $A_k^{(j)}$, and $B_k^{(j)}$
    for all $(j,k)\in\mathcal I_q$
    \State Initialize $w_0\gets\mathbf0$

    \For{$r=0,\dots,r_{\max}-1$}
        \State Determine
        $\mathcal F_{k,r}^{(j)}$, $\mathcal S_{k,r}^{(j)}$,
        and the saturation signs from $w_r$

        \State Solve the fixed-pattern controller problem
        \Statex \hspace{\algorithmicindent}
        $\displaystyle
        w_{r+1}\in
        \arg\min_w\ell_w\|w\|^2
        \quad
        \text{s.t. }\eqref{eq:c6},\eqref{eq:c5_active}$

        \If{\eqref{eq:c6} and \eqref{eq:c5_active} admit no solution}
            \State \textbf{break}
        \EndIf

        \If{\eqref{eq:c5} holds on $\mathcal I_q$ using
        $\operatorname{sat}_{\bar u}
        (u_{d_k}^{(j)}(w_{r+1}))$}
            \State \textbf{return} $(\theta_q^\star,w_{r+1})$
        \EndIf

        \If{the sets and signs induced by $w_{r+1}$ are unchanged}
            \State \textbf{break}
        \EndIf
    \EndFor
\EndFor

\State \textbf{return no verified candidate found}
\end{algorithmic}
\end{algorithm}

\begin{rem}
The finite search over $v$ and the active-set iteration
limit ensure termination, but not global optimality or an infeasibility
certificate for \eqref{eq:joint_problem}. For a fixed saturation
pattern, the funnel and controller steps are, respectively, a linear
program and a convex quadratic program, whose sizes scale mainly with
$n_c$, $N_K$, and $N N_t$. Actuator bounds hold by construction through
saturation; directly imposing $|u_{d,i}|\leq\bar u_i$ as a hard constraint in \eqref{eq:joint_problem}, would be more
conservative and could exclude candidates whose saturated inputs satisfy
the funnel-invariance condition \eqref{eq:c5}.
\end{rem}

\section{Theoretical Guarantees}
\label{sec:guarantees}
In this section, we establish formal safety and performance guarantees for the proposed framework. In particular, we present two complementary certificates: a conservative semi-global result based on actuator authority, and a data-driven local refinement.
\vspace{-1mm}
\subsection{Baseline Safety Certificate}
\label{sec:baseline_cert}
For a feasible pair $(\theta^\star,w^\star)$, define
$\underline k^\star:=\inf_{\xi\in(-1,1)^n}
\lambda_{\min}(K(\xi;w^\star))$,
$\bar k^\star:=\sup_{\xi\in(-1,1)^n}
\lambda_{\max}(K(\xi;w^\star))$, and
$s_{\mathrm{th}}^\star:=
(\bar\delta_f+\bar\delta_G\|\bar u\|)/\underline k^\star$.
For an initial condition inside the learned funnel, let
$\underline\rho^\star:=\min_{i,t}\rho_i(t;\theta^\star)$,
$\bar\rho^\star:=\max_{i,t}\rho_i(t;\theta^\star)$,
$\bar E^\star:=\max\{\|\varepsilon(0)\|,
\bar\rho^\star s_{\mathrm{th}}^\star\}$, and
$s_{\max}^\star:=
\frac{\bar E^\star\cosh^2(\bar E^\star)}
{\underline\rho^\star}$,
where the extrema are over $i=1,\dots,n$ and $t\in[0,T]$.
The following result provides a conservative semi-global certificate for funnel invariance under a sufficient control authority condition.

\begin{theorem}
\label{lem:ppc_certificate}
Consider \eqref{eq:plant} under Assumption~\ref{ass:model}. Let
$(\theta^\star,w^\star)$ be feasible for \eqref{eq:joint_problem}, suppose
that $|x_i(0)|<\rho_i(0;\theta^\star)$ for all $i=1,\dots,n$ and define $\lambda :=
\bigl[
\lambda_1,\dots,
\lambda_n
\bigr]^\top$. If
\begin{equation}
\bar u_i
\geq
\underline g^{-1}
\left(
v\|\lambda\|+\bar F_{\hat f}
+\bar k^\star s_{\max}^\star
\right),
\qquad i=1,\dots,n
\label{controlcond}
\end{equation}
then the closed-loop solution exists on $[0,T]$ and satisfies
$|x_i(t)|<\rho_i(t;\theta^\star)$ and
$|u_i(t,x(t);\theta^\star,w^\star)|\leq\bar u_i$
for all $t\in[0,T]$ and $i=1,\dots,n$.
\end{theorem}
\begin{proof}
For notational simplicity, we omit the explicit dependence on $\theta^\star$ and write $\rho(t)$, $\xi(t)$, $\varepsilon(t)$, $R(t)$, $s(t)$, $q(t,x)$, $u_d(t,x;w^\star)$ and $u(t,x;w^\star)$ whenever no ambiguity arises. Let $\tau_{\max}\in(0,T]$ denote the maximal time such that
$\xi(t)\in\Omega_\xi$ for all $t\in[0,\tau_{\max})$. By Theorem~\ref{lem:slope_uniform}, the constraints \eqref{eq:c1}-\eqref{eq:c4} imply $\rho_i(t)>0$ on $[0,T]$ and the initial condition satisfies
$|\xi_i(0)|<1$. Hence, on this interval, the transformed error
$\varepsilon_i=\atanh(\xi_i)$ is well defined. Next, consider the open set $\Omega_\xi :=
\left\{\xi \in \mathbb{R}^n:
\xi_i\in(-1,1),\;
i=1,\dots,n\right\}$. Exploiting Theorem 54 in \cite{ds} (pp. 476) we conclude the existence and uniqueness of a maximal solution $\xi:[0,\tau_{\max})\to \Omega_\xi$ such that
$\xi(t)\in\Omega_\xi$ for all $t\in[0,\tau_{\max})$.

Consider the Lyapunov function $V(\varepsilon)=\tfrac12 \varepsilon^\top\varepsilon$. 
Consider also the maximal subinterval of $[0,\tau_{\max})$ on which saturation is inactive, so that
$u(t,x;w^\star)=u_d(t,x;w^\star)$. Then, from \eqref{eq:Vdot_bound_sat}:
\begin{equation*}
\dot V
\le
s^\top\bigl(\hat f(x)+\hat G(x)u_d(t,x;w^\star)-q(t,x)\bigr)
+
\|s\|\bigl(\bar\delta_f+\bar\delta_G\|\bar u\|\bigr).
\end{equation*}
Substituting \eqref{desiredu} and since $K(\xi;w^\star)\succeq \underline k^\star I$, we have:
\begin{equation}
\dot V
\le
-\underline k^\star\|s\|^2
+
\|s\|\bigl(\bar\delta_f+\bar\delta_G\|\bar u\|\bigr).
\label{eq:vdot_unsat_end}
\end{equation}
Hence, $\dot V\leq0$ whenever $\|s\|\geq s_{\mathrm{th}}^\star$. Since
$\|s\|\geq\|\varepsilon\|/\bar\rho^\star$, it follows that
$\dot V\leq0$ whenever
$\|\varepsilon\|\geq\bar\rho^\star s_{\mathrm{th}}^\star$.
Therefore,
$\|\varepsilon(t)\|\leq\bar E^\star$ on $[0,\tau_{\max})$.
Since $|\xi_i(t)|\leq\tanh(\bar E^\star)$ we have $\underline\rho^\star \|R(t)\|
\leq
\cosh^2(\bar E^\star).$
Therefore:
\[
\|s(t)\|
\leq
\|R(t)\|\,\|\varepsilon(t)\|
\leq
\frac{\bar E^\star\cosh^2(\bar E^\star)}
{\underline\rho^\star}
=s_{\max}^\star.
\]
Moreover, since $|x_i(t)|<\rho_i(t)$ on $[0,\tau_{\max})$ and by leveraging Theorem~\ref{lem:slope_uniform} we get $|q_i(t,x)|\le |\dot\rho_i(t)| \le  v\lambda_i$ for all
$i=1,\dots,n$. Therefore, it holds that $\|q(t,x)\|\le v\|\lambda\|$. By the containment assumption, $x(t)\in\mathcal X_c$ on $[0,\tau_{\max})$. Using the assumption $\|\hat G(x)^{-1}\|\le \underline g^{-1}$ and $\|\hat f(x)\|\le \bar F_{\hat f}$ on $\mathcal X_c$, we conclude that:
\[
\|u_d(t)\|
\leq
\underline g^{-1}
\left(
v\|\lambda\|+\bar F_{\hat f}
+\bar k^\star s_{\max}^\star
\right).
\]
Condition~\eqref{controlcond} implies
$|u_{d,i}(t)|\leq\bar u_i$ on this subinterval, thus saturation cannot become active on $[0,\tau_{\max})$.
Next define the set $\Omega'_\xi :=
\left\{\xi \in \mathbb{R}^n:
|\xi_i|\leq \tanh(\bar E^\star),\;
i=1,\dots,n\right\}$
Since $\xi(t) \in \Omega'_\xi \subset \Omega_\xi$, standard continuation arguments imply that the solution cannot leave
$\Omega_\xi$ before $T$ and therefore extends to the entire interval $[0,T]$, concluding the proof.
\end{proof}

\begin{rem}
Condition~\eqref{controlcond} is sufficient and may be conservative. It
uses a uniform norm bound for all desired-input components over the
entire funnel and requires $u_d$ to remain within the actuator limits,
so that saturation is inactive and $u=u_d$. Less conservative
certificates could use actuator-specific bounds over the relevant
funnel region or verify Lyapunov decrease for individual saturation
patterns.
\end{rem}
\vspace{-1mm}
\subsection{Data-Driven Local Certification}
\label{sec:local_cert}

We now give a local certificate under explicit regularity and coverage assumptions for strict Lyapunov decrease near the funnel boundary. 
\begin{theorem}
\label{prop:local}
Consider \eqref{eq:plant} under Assumption~\ref{ass:model}, and let
$(\theta^\star,w^\star)$ be feasible for \eqref{eq:joint_problem}. Define $\mathcal X_\delta
:=
\{
(t,x)\in[0,T]\times\mathbb R^n:
|x_i(t)|<\rho_i(t;\theta^\star),\ i=1,\dots,n,\ 
\|s(t,x)\|\geq\delta
\}.$ and assume that $\mathcal X_\delta\subseteq[0,T]\times\mathcal X_c$. Define $V(\varepsilon)=\tfrac12 \varepsilon^\top\varepsilon$ and assume that
$\dot V(\varepsilon(t,x))$ is $L$-Lipschitz on $\mathcal X_\delta$ with
respect to $|t-t'|+\|x-x'\|$. Suppose also that, for every
$(t,x)\in\mathcal X_\delta$, there exists a data point
$x_k^{(j)}$, with $j=1,\dots,N$ and $k=1,\dots,N_t$, such that $|t-t_k|\leq h$ and $\|x-x_k^{(j)}\|\leq\varrho_{\mathcal D}$.
If $L(h+\varrho_{\mathcal D})<\mu$, then, for every initial condition
satisfying $|x_i(0)|<\rho_i(0;\theta^\star)$, the closed-loop
solution exists on $[0,T]$ and satisfies
\[
|x_i(t)|<\rho_i(t;\theta^\star),
\qquad
|u_i(t,x(t);\theta^\star,w^\star)|\leq\bar u_i
\]
for all $t\in[0,T]$ and $i=1,\dots,n$.
\end{theorem}

\begin{proof}
Fix any $(t,x)\in\mathcal X_\delta$. By the coverage assumption, select
a sampled point $(t_k,x_k^{(j)})\in\mathcal X_\delta$ such that
\[
|t-t_k|\leq h,
~
\|x-x_k^{(j)}\|\leq\varrho_{\mathcal D}.
\] Since
$\|s_k^{(j)}\|\geq\delta$, we have
$b_\delta(\|s_k^{(j)}\|)=1$. Therefore, feasibility of
\eqref{eq:joint_problem}, together with
\eqref{eq:Vdot_bound_sat} and \eqref{eq:c5}, gives:
\[
\dot V(\varepsilon_k^{(j)})
\leq
A_k^{(j)}\operatorname{sat}_{\bar u}(u_{d_k}^{(j)})
+B_k^{(j)}
\leq-\mu.
\]
Lipschitz continuity then yields:
\begin{align*}
\dot V(\varepsilon(t,x))
&\leq
-\mu+L\bigl(|t-t_k|+\|x-x_k^{(j)}\|\bigr)\\
&\leq
-\mu+L(h+\varrho_{\mathcal D})<0.
\end{align*}
Hence, $\dot V<0$ for all $(t,x)\in\mathcal X_\delta$.

Next, consider any trajectory with $|x_i(0)|<\rho_i(0;\theta^\star)$. Suppose, by contradiction, that there exists a first
$t_b\in(0,T]$ such that $|\xi_i(t)|\to1$ as $t\uparrow t_b$
for some $i$. Then $|\varepsilon_i(t)|=|\atanh(\xi_i(t))|\to\infty$ as $t\uparrow t_b$, so $V(\varepsilon(t))\to\infty$. Moreover, since $s_i(t)=\varepsilon_i(t)/(\rho_i(t)(1-\xi_i^2(t)))$, it follows that $\|s(t)\|\to\infty$, and thus there exists $t_\delta<t_b$ such that $\|s(t)\|\ge\delta$ for all $t\in[t_\delta,t_b)$, i.e., $(t,x(t))\in\mathcal X_\delta$. Hence $\dot V(\varepsilon(t))<0$ on $[t_\delta,t_b)$, so $V(\varepsilon(t))$ is decreasing there and remains bounded, contradicting $V(\varepsilon(t))\to\infty$. Therefore $|\xi_i(t)|<1$ for all $t\in[0,T]$, which implies $|x_i(t)|<\rho_i(t;\theta^\star)$. Input constraints follow from the saturation in \eqref{eq:controller}, completing the proof.
\end{proof}

\begin{rem}
Unlike Theorem~\ref{lem:ppc_certificate}, this certificate replaces the uniform actuator-authority condition \eqref{controlcond} with Lipschitz regularity and sufficiently dense active demonstrations near the funnel boundary. Feasibility of
\eqref{eq:joint_problem} verifies Lyapunov decrease only at finitely many
sampled points; the regularity and coverage conditions are needed to
extend this guarantee to the surrounding time-state region and
establish funnel invariance.
\end{rem}
\section{Simulation Results}
We consider a nonlinear coupled-tank system of the form \eqref{eq:plant} with:
\[
f(x)=\begin{bmatrix}
-\frac{\alpha_1}{A_1}\sqrt{x_1 + h_1^\star}
-\frac{\alpha_3}{A_1}\sqrt{x_1 - x_2 + h_1^\star - h_2^\star} \\
-\frac{\alpha_2}{A_2}\sqrt{x_2 + h_2^\star}
+\frac{\alpha_3}{A_2}\sqrt{x_1 - x_2 + h_1^\star - h_2^\star}
\end{bmatrix}
\]
and constant input matrix $G=\mathrm{diag}(1/A_1,\,1/A_2)$. The state is denoted by $h = [h_1\;h_2]^\top$, where $h_i$ is the liquid level in tank $i$, and
$x=h-h^\star$ denotes its deviation from the setpoint
$h^\star=[0.8\;0.4]^\top$. Nominal parameters are $A_1=A_2=1$, $\alpha_1=\alpha_2=0.5$, $\alpha_3=0.3$, while the true system is given as $A_1=1.03$, $A_2=0.97$, $\alpha_1=0.54$, $\alpha_2=0.46$, $\alpha_3=0.33$, i.e., $\bar\delta_f = 0.05$ and $\bar\delta_G = 0.031$.

We generate $N=30$ trajectories using heterogeneous PPC controllers $u_{\mathrm{exp}}^{(j)}(t,x)=\mathrm{sat}_{\bar u}\bigl(-K_j s^{(j)}(t,x)\bigr)$, with $K_j=\mathrm{diag}(k_{j,1},k_{j,2})$, $k_{j,1},k_{j,2}\in[8,20]$ and $s^{(j)}$ is computed as in Section~\ref{PPC} with the prescribed exponential funnels $\rho_i^{(j)}(t)
=
\bigl(\rho_i(0)-\rho_{\infty,i}^{(j)}\bigr)
e^{-l_i^{(j)}t}
+\rho_{\infty,i}^{(j)}$, where $l_i^{(j)}\in[6,14]$, $\rho_{\infty,1}^{(j)}\in[0.03,0.12]$, $\rho_{\infty,2}^{(j)}\in[0.02,0.10]$,~ $j=1,\dots,30$, $i=1,2$, while $\bar u=[6,\,5]^\top$ and $\rho(0)=[0.55,\,0.45]^\top$ are fixed. We deliberately use existing PPC controllers to
generate the expert trajectories in this example. We make this choice simply to generate demonstrations with desirable
behavior; in practice, the state-trajectory data could be collected
using any suitable controller or through expert manual operation. Trajectories are simulated over $[0,T]$, $T=5~s$, from admissible initial conditions satisfying $h_1(0) > h_2(0) > 0$, sampled at $N_t=281$. Empirical quantiles of $|x_i|$ are computed with $\beta=0.95$ and $\alpha=0.1$, $\lambda_i=10$, and are smoothed via a moving average (window $5$). Funnels are parameterized by cubic B-splines ($n_c=250$, $\underline{\rho}_i=10^{-4}$, $c_i=40$), and gains via Gaussian RBFs (28 centers in $[-0.95,0.95]$, $\sigma=0.15$, $k_{\min}=0.1$). The parameters are $\mu=0.1$, $\delta=3$. The synthesis uses $\ell_w=10^{-5}$, $v=1$, $\gamma_v=0.8$, $v_{\min}=0.08$, and $r_{\max}=10$. Fig.~\ref{fig:coupled_tank_main} summarizes the learned funnel–controller pair. The learned funnels remain within the empirical quantile envelopes (regulated by $\alpha,\beta$), showing that the proposed method captures the transient and steady-state behavior of heterogeneous demonstrations. 
\begin{figure}[t]
\centering
\includegraphics[width=1\linewidth,trim={0.0cm 0.0cm 0cm 0.0cm},clip]{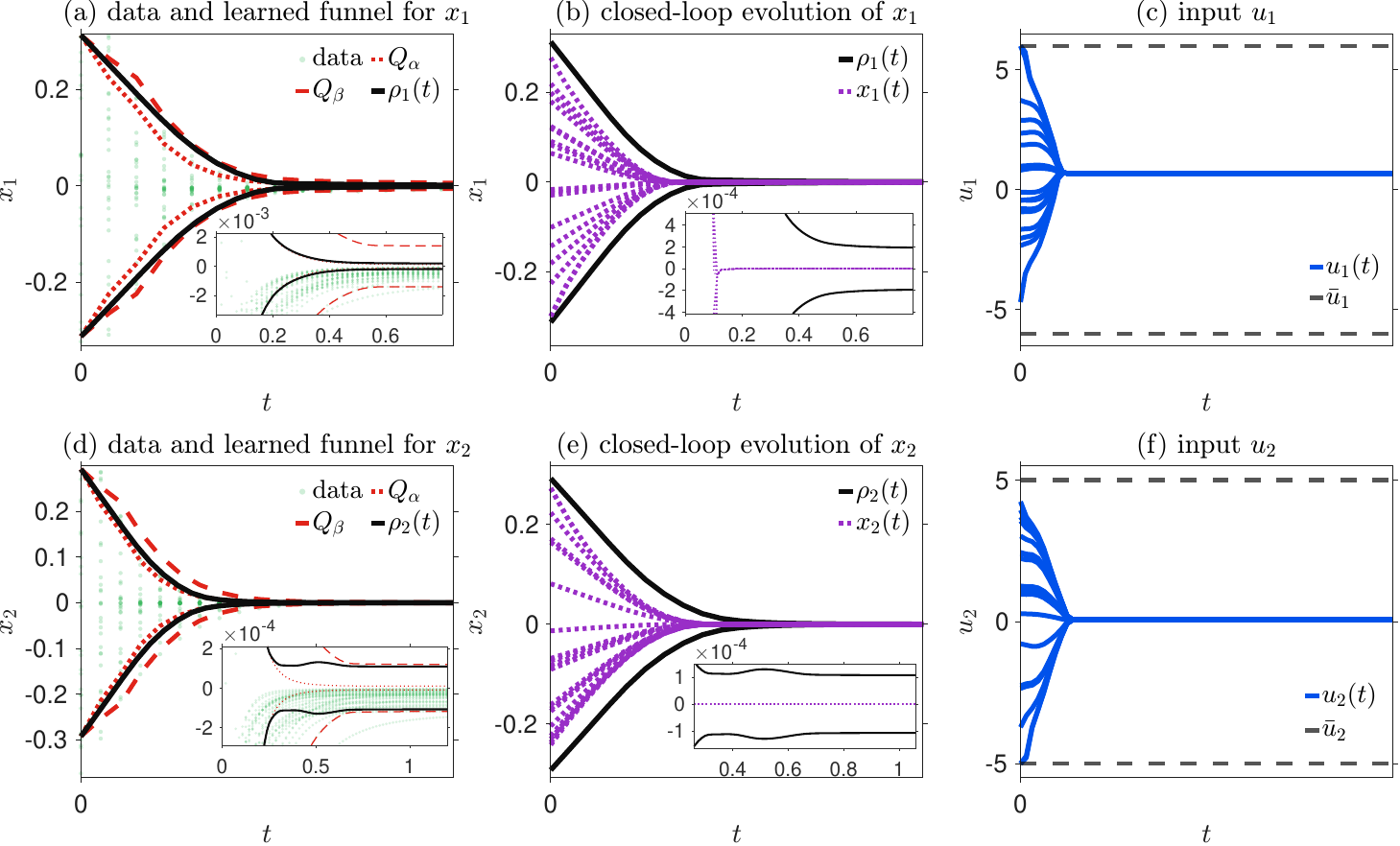}
\caption{Coupled-tank system. Left: data, $\alpha/\beta$-quantiles, and learned funnels. Middle: closed-loop trajectories for $20$ unseen initial conditions. Right: inputs with bounds.}
\label{fig:coupled_tank_main}
\end{figure}

For $20$ unseen initial conditions, the closed-loop trajectories remain strictly inside the learned funnels, showing empirical success on unseen initial conditions. At the same time, the control inputs satisfy the actuator bounds for all simulations. Notably, although the demonstrations are generated by multiple PPC controllers with different gains and funnels, the proposed method learns a single closed-form controller that mimics their behavior within the user-selected quantile bounds defined by $\alpha,\beta$. Since $f(0)\neq0$, the origin is not an open-loop equilibrium. Thus,
without drift compensation, model-free PPC generally achieves only
practical regulation, as shown by the green trajectories in the zoomed plots of Figs.~\ref{fig:coupled_tank_main}(a) and (d). The proposed
model-based controller compensates for this drift and respects the input
limits while achieving exponential convergence to the origin in this example.

\section{Conclusion}
In this work, we studied the problem of learning actuator-feasible prescribed performance controllers from state trajectory data for nonlinear systems with input constraints. We proposed a data-driven framework that jointly synthesizes a time-varying performance funnel and a structured PPC controller from state-only data. The resulting synthesis problem is addressed via a feasibility-driven active-set synthesis procedure that explicitly accounts for actuator saturation. Moreover, we establish two complementary guarantees ensuring prescribed performance and input constraint satisfaction: a conservative semi-global certificate based on actuator authority, and a local data-driven certificate under sufficient data density. Future work includes reducing the conservativeness of the actuator conditions, and extending the framework to broader system classes and richer controller parameterizations.
\bibliographystyle{IEEEtran}
\bibliography{IEEEabrv,references}

\begin{thebibliography}{10}
\providecommand{\url}[1]{#1}
\csname url@samestyle\endcsname
\providecommand{\newblock}{\relax}
\providecommand{\bibinfo}[2]{#2}
\providecommand{\BIBentrySTDinterwordspacing}{\spaceskip=0pt\relax}
\providecommand{\BIBentryALTinterwordstretchfactor}{4}
\providecommand{\BIBentryALTinterwordspacing}{\spaceskip=\fontdimen2\font plus
\BIBentryALTinterwordstretchfactor\fontdimen3\font minus \fontdimen4\font\relax}
\providecommand{\BIBforeignlanguage}[2]{{%
\expandafter\ifx\csname l@#1\endcsname\relax
\typeout{** WARNING: IEEEtran.bst: No hyphenation pattern has been}%
\typeout{** loaded for the language `#1'. Using the pattern for}%
\typeout{** the default language instead.}%
\else
\language=\csname l@#1\endcsname
\fi
#2}}
\providecommand{\BIBdecl}{\relax}
\BIBdecl

\bibitem{ppc}
C.~P. Bechlioulis and G.~A. Rovithakis, ``Robust adaptive control of feedback linearizable mimo nonlinear systems with prescribed performance,'' \emph{IEEE Transactions on Automatic Control}, vol.~53, no.~9, pp. 2090--2099, 2008.

\bibitem{bechlaut}
C.~P. Bechlioulis and G.~A. Rovithakis, ``A low-complexity global approximation-free control scheme with prescribed performance for unknown pure feedback systems,'' \emph{Automatica}, vol.~50, no.~4, pp. 1217--1226, 2014.

\bibitem{bergerA}
T.~Berger, H.~H. L\^e, and T.~Reis, ``Funnel control for nonlinear systems with known strict relative degree,'' \emph{Automatica}, vol.~87, pp. 345--357, 2018.

\bibitem{paper1}
P.~K. Mishra and P.~Jagtap, ``Approximation-free prescribed performance control with prescribed input constraints,'' \emph{IEEE Control Systems Letters}, vol.~7, pp. 1261--1266, 2023.

\bibitem{smc}
P.~S. Trakas and C.~P. Bechlioulis, ``Adaptive performance control for input constrained mimo nonlinear systems,'' \emph{IEEE Transactions on Systems, Man, and Cybernetics: Systems}, vol.~54, no.~12, pp. 7733--7745, 2024.

\bibitem{ratesiso}
P.~S. Trakas and C.~P. Bechlioulis, ``Robust adaptive prescribed performance control for unknown nonlinear systems with input amplitude and rate constraints,'' \emph{IEEE Control Systems Letters}, vol.~7, pp. 1801--1806, 2023.

\bibitem{bergerArch}
T.~Berger, ``Input-constrained funnel control of nonlinear systems,'' \emph{IEEE Transactions on Automatic Control}, vol.~69, no.~8, pp. 5368--5382, 2024.

\bibitem{fotiadis}
F.~Fotiadis and G.~A. Rovithakis, ``Input-constrained prescribed performance control for high-order mimo uncertain nonlinear systems via reference modification,'' \emph{IEEE Transactions on Automatic Control}, vol.~69, no.~5, pp. 3301--3308, 2024.

\bibitem{autom}
L.~N. Bikas and G.~A. Rovithakis, ``Prescribed performance under input saturation for uncertain strict-feedback systems: A switching control approach,'' \emph{Automatica}, vol. 165, p. 111663, 2024.

\bibitem{ames}
A.~D. Ames, X.~Xu, J.~W. Grizzle, and P.~Tabuada, ``Control barrier function based quadratic programs for safety critical systems,'' \emph{IEEE Transactions on Automatic Control}, vol.~62, no.~8, pp. 3861--3876, 2017.

\bibitem{lanza}
L.~Lanza, J.~Köhler, D.~Dennstädt, T.~Berger, and K.~Worthmann, ``A model-free approach to control barrier functions using funnel control,'' \emph{IEEE Control Systems Letters}, vol.~9, pp. 1183--1188, 2025.

\bibitem{verginis}
C.~K. Verginis, ``Funnel control for uncertain nonlinear systems via zeroing control barrier functions,'' \emph{IEEE Control Systems Letters}, vol.~7, pp. 853--858, 2023.

\bibitem{namerikawa}
R.~Namerikawa, A.~Wiltz, F.~Mehdifar, T.~Namerikawa, and D.~V. Dimarogonas, ``On the equivalence between prescribed performance control and control barrier functions,'' in \emph{2024 American Control Conference (ACC)}, 2024, pp. 2458--2463.

\bibitem{cbfl}
A.~Robey, H.~Hu, L.~Lindemann, H.~Zhang, D.~V. Dimarogonas, S.~Tu, and N.~Matni, ``Learning control barrier functions from expert demonstrations,'' in \emph{2020 59th IEEE Conference on Decision and Control (CDC)}, 2020, pp. 3717--3724.

\bibitem{jin2020neural}
W.~Jin, Z.~Wang, Z.~Yang, and S.~Mou, ``Neural certificates for safe control policies,'' \emph{arXiv preprint arXiv:2006.08465}, 2020.

\bibitem{dawson2022safe}
C.~Dawson, Z.~Qin, S.~Gao, and C.~Fan, ``Safe nonlinear control using robust neural lyapunov-barrier functions,'' in \emph{Proceedings of the 5th Conference on Robot Learning}, 2022, pp. 1724--1735.

\bibitem{torabi2018behavioral}
F.~Torabi, G.~Warnell, and P.~Stone, ``Behavioral cloning from observation,'' in \emph{Proceedings of the 27th International Joint Conference on Artificial Intelligence}, ser. IJCAI'18.\hskip 1em plus 0.5em minus 0.4em\relax AAAI Press, 2018, p. 4950–4957.

\bibitem{splines}
C.~de~Boor, ``A practical guide to splines,'' in \emph{Applied Mathematical Sciences}, 1978.

\bibitem{ds}
E.~D. Sontag, \emph{Mathematical Control Theory: Deterministic Finite Dimensional Systems (2nd Ed.)}.\hskip 1em plus 0.5em minus 0.4em\relax Berlin, Heidelberg: Springer-Verlag, 1998.

\end{thebibliography}
\end{document}